\documentclass[runningheads]{llncs}
\usepackage[T1]{fontenc}
\usepackage{graphicx}
\usepackage{cite} 
\usepackage{wrapfig} 
\usepackage{todonotes}
\usepackage{xspace}
\usepackage{multicol}
\usepackage{appendix}
\usepackage{xcolor}
\usepackage{hyperref}
\usepackage{authblk}
\usepackage[normalem]{ulem}
\usepackage{bbding}
\usepackage{soul}
\usepackage[linesnumbered,ruled,vlined]{algorithm2e}
\usepackage{array}
\usepackage{enumitem}

\newcommand{\cmnt}[1]{}
\newcommand{\ignore}[1]{}
\newcommand{\remove}[1]{}

\newcommand{\ceil}[1]{\lceil #1 \rceil}

\newcommand {\spcolor}[1] {\textcolor{black}{#1}}

\newcommand {\mpcolor}[1] {\textcolor{black}{#1}}

\newtheorem{observation}[lemma]{Observation}

\newcommand{\secref}[1]{Section~\ref{sec:#1}}
\newcommand{\figref}[1]{Fig.~\ref{fig:#1}}

\newcommand{\thmref}[1]{Theorem~\ref{thm:#1}}

\newcommand{\lineref}[1]{Line~\ref{lin:#1}}
\newcommand{\algoref}[1]{{Algorithm \ref{alg:#1}}}

\newcommand{\Lineref}[1]{Line~\ref{lin:#1}}

\newcommand{\wf} {wait-free\xspace}

\newcommand{\CAS} {\textit{CAS}\xspace}

\newcommand{\cas} {CAS\xspace}

\newcommand{\Fwrite}{{\tt StickyCAS-append}\xspace}
\newcommand{\Fview}{{\tt StickyCAS-read}\xspace}
\newcommand{\BVCAS}{StickyCAS\xspace}
\newcommand{\bvcas}{StickyCAS\xspace}
\newcommand{\abcas}{StickyCAS\xspace}

\newcommand{\ie}{\textit{i.e., }}

\newcommand{\mth} {method\xspace}

\newcommand{\sbc} {Strong Byzantine Consensus\xspace}
\newcommand{\cvbc} {Common-value Byzantine Consensus\xspace}
\newcommand{\wbc} {Weak Byzantine Consensus\xspace}

\begin{document}
\title{Byzantine-Tolerant Consensus in GPU-Inspired Shared Memory}
\author{Chryssis Georgiou\inst{1} \and
Manaswini Piduguralla\inst{2} \Envelope \and
Sathya Peri\inst{2}}

\institute{
    University of Cyprus, Cyprus
    \email{chryssis@ucy.ac.cy }
    \and
    Indian Institute of Technology Hyderabad, India
    \email{\{cs20resch11007@,sathya\_p@cse.\}iith.ac.in} 
}\authorrunning{C. Georgiou et al.}
\titlerunning{Byzantine Consensus in GPU}
\maketitle

\begin{abstract}
In this work, we formalize a novel shared memory model inspired by the popular GPU architecture. Within this model, we develop algorithmic solutions to the Byzantine Consensus problem and analyze their fault-resilience. 

\end{abstract}

\section{Introduction}
\label{sec:Intro}
In modern computing, Graphics Processing Units (GPUs) have transcended their traditional role in graphics rendering, emerging as powerful platforms for parallel computation. Unlike conventional CPUs, which have been widely studied, GPUs offer a fundamentally different architectural approach. We believe that leveraging GPU-based models can provide innovative solutions to critical distributed computing challenges, including consensus, leader election, and atomic broadcast. In this paper, we introduce a GPU-inspired computational model and demonstrate its effectiveness in addressing the Byzantine consensus problem. To establish a foundational understanding of GPU architecture, we first present key concepts. 

\noindent
\textbf{The GPU Architecture.} GPU is a specialized processor originally developed to meet the demands of the rapidly expanding video game industry. Designed to execute a large number of floating point calculations and memory operations per video frame in advanced games, GPUs have since found widespread applications in scientific computing, artificial intelligence (AI), and high performance computing (HPC)\cite{david+:PMPP:science:2017}. Unlike traditional Central Processing Units (CPUs), which excel at managing complex control logic, GPUs are optimized for efficiently handling thousands of simple parallel tasks. 
\begin{figure}[t]
\centering
\includegraphics[scale=0.46]{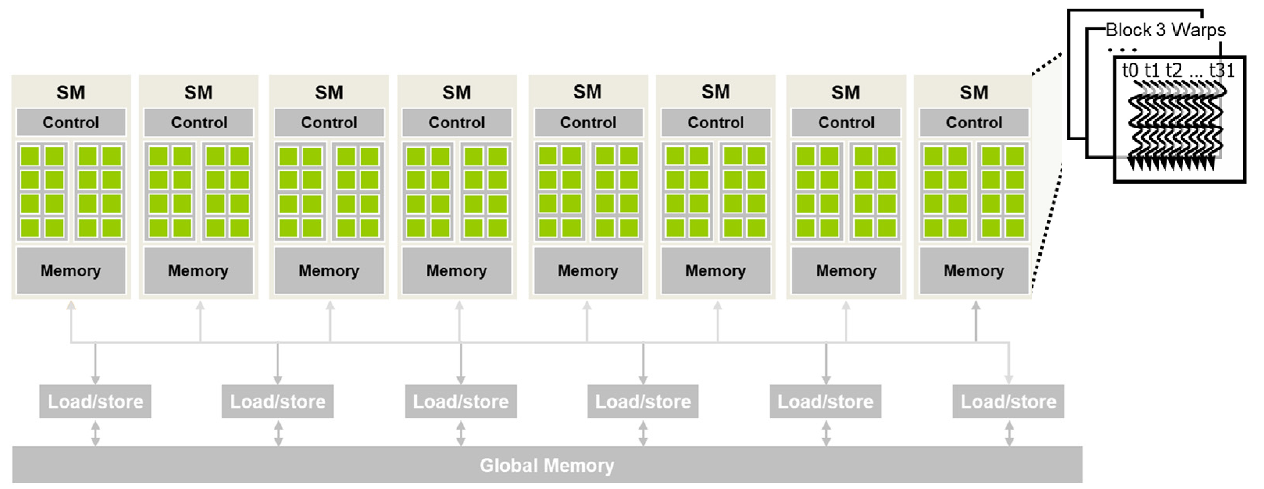}
\caption{Illustration of CUDA-capable GPU Architecture\cite{david+:PMPP:science:2017}}
\label{fig:gpuArch}
\end{figure}

\figref{gpuArch} presents a high-level overview of a typical GPU architecture. The architecture is built around \emph{Streaming Multiprocessors (SMs)}, which are the fundamental computational units. Each SM comprises multiple processing units, known as \emph{streaming processors} or \emph{cores}. However, these cores lack individual program counters, unlike CPU cores. Within an SM, the computation is further organized into {\bf\em thread blocks}, where each block consists of a group of threads that collaborate and share local resources, such as registers and shared memory. Execution within an SM occurs in groups of (currently) 32 threads, 
referred to as {\bf\em  warps}, which operate under the \textbf{Single Instruction, Multiple Threads (SIMT)} execution model. In this model, all threads within a warp execute the same instructions simultaneously but possibly on different data. 
Additionally, suppose that the total number of active threads in a warp is fewer than the number of available cores. In that case, the unused cores will remain idle, underutilizing the GPU's processing power. Thus, {\em threads within a warp execute synchronously} in a lock-step fashion, whereas {\em warps within a block execute asynchronously} and may proceed in any order relative to one another. The hierarchy of GPU components and the execution model is illustrated in \figref{GPU-hier}.

\begin{wrapfigure}{r}{0.5\textwidth} 
    \centering
    \includegraphics[scale=0.3]{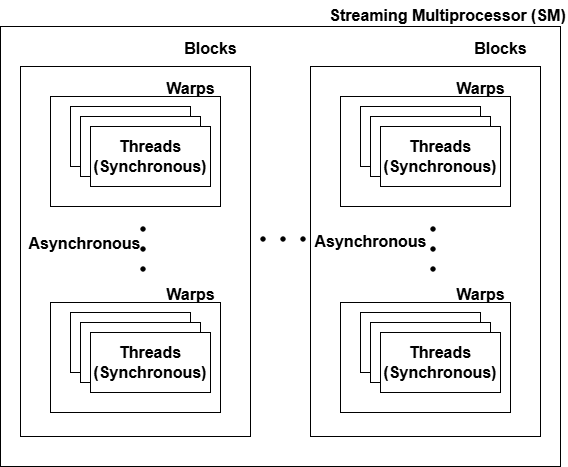}
    \caption{Illustration of GPU Hierarchy}
    \label{fig:GPU-hier}
\end{wrapfigure}Registers and shared memory are on-chip memories, and registers are allocated to individual threads; each thread can access only its own registers \cite{david+:PMPP:science:2017}. Shared memory is allocated to thread blocks; all threads in a block can access shared memory variables declared for the block.  Some GPU implementations (example: NVIDIA) come with their own hardware internal scheduler, called {\bf\em warp scheduler}, which allocates the warps to whatever GPU hardware is present \cite{Olmedo+:RTAS:IEEE:2020}. Each warp executes instructions in cycles, with the warp scheduler selecting which warps to execute based on resource availability. In each cycle, a warp executes a constant (e.g, one or two) number of instructions~\cite{Ogier:book:2013,NVIDIA:2024}. This scheduling approach is designed to mask the latency that warps experience while preparing the next instruction for execution.

GPUs and CPUs have fundamentally different architectures which are optimized for distinct workloads. CPUs are designed to minimize instruction execution latency, whereas GPUs are built to maximize throughput \cite{david+:PMPP:science:2017}. As discussed earlier, threads within a warp execute synchronously, while execution across warps is asynchronous. The warp scheduler also ensures progress for the threads, which otherwise can be a major challenge in designing solutions for asynchronous processes. Additionally, constructing shared memory with varying read-write access across multiple threads is not as straightforward as in CPUs. 




\noindent
\textbf{Byzantine Faults in GPU Systems.} In modern computing environments, GPUs are increasingly shared across multiple processes, such as in NVIDIA’s Multi-Process Service (MPS), requiring a mechanism for fair, efficient and fault-tolerant resource allocation \cite{Weaver+:SC24:2024}. Furthermore, in HPC clusters, where multiple GPUs collaborate on complex computations, failures in individual nodes can jeopardize the integrity of the entire task. Another key motivation stems from the well-documented observation that GPUs exhibit a higher susceptibility to hardware errors compared to CPUs \cite{CiniYalcin:GPU-CPU:ACMSurvey:2020}.

 These hardware faults can manifest as arbitrary or erroneous software behavior, which can be particularly challenging to predict and mitigate. Transient faults, such as bit flips caused by cosmic rays or voltage fluctuations, are notoriously difficult to model. We argue that a Byzantine fault-tolerant (BFT) approach is well-suited for handling such unpredictable failures, as it can provide correctness guarantees even in the presence of arbitrary faults.


\noindent
\textbf{Motivation for Byzantine Consensus in GPUs.} Agreement among multiple parties is a fundamental problem in distributed computing  for a wide range of applications \cite{kshemkalyani:book:2011}. Considerable research efforts have been directed toward delving into consensus mechanisms within shared memory systems. Currently, there is an increased emphasis on Byzantine fault-tolerant shared objects \cite{CKDISC2021,KMPOPODIS2023}.\mpcolor{ This interest is further motivated by the universality of consensus \cite{Maurice:1991:ACMTrans}, as solving consensus enables the wait-free implementation of all shared objects in asynchronous systems.}

 A robust consensus protocol can ensure priority-aware scheduling, prevent resource starvation, and optimize performance in multi-tenant GPU systems \mpcolor{ \cite{Kolesnichenko+:2017+GPCE}}. Thus, considering the popularity of GPU systems and the Byzantine errors that can possibly be exhibited by these systems (as argued above), in this work we explore Byzantine Fault-tolerant Consensus algorithms in GPU Shared memory systems.


\noindent
\textbf{Contributions.} The main contributions of this paper are the following:\break (a) We formalize a novel GPU-inspired Shared Memory Model that abstracts the unique features of the GPU architecture  detailed above (\secref{model}).\break (b) We detail the Byzantine Consensus problems we consider in this work (Section~\ref{sec:ByzIntro}) and develop a Byzantine-tolerant consensus solution tailored for this memory model, thus demonstrating its utility (\secref{solns}). Our solution utilizes the \BVCAS object, a novel shared object we introduce and specify. In Appendix~\ref{app}, we provide a \BVCAS implementation which is built using CAS.\break (c) \mpcolor{We prove }the solution's correctness and fault-resilience (Section~\ref{sec:correctness}). 


\section{The GPU-inspired Shared Memory Model}
\label{sec:model}
In this section, we formalize a shared memory model inspired by the GPU framework detailed in Section~\ref{sec:Intro}. We consider a shared memory system where processes communicate only by accessing shared memory objects. Following Cohen and Keidar~\cite{CKDISC2021}, we assume that the shared memory is reliable (e.g., it cannot become inaccessible or corrupted), which is also commonly assumed in HPC Applications~\cite{Huang+:SDC-HPC:SC:2022}. The proposed model is depicted  in \figref{SysModel} and the various components are detailed here. 

\ignore{
\noindent{\bf GPU Execution Model.} In this work, we are inspired by General Purpose Graphics Processing Unit (GPU) model of semi synchronous execution. As detailed in \secref{Intro}, the architecture of the GPUs are hierarchically categorized as cores, warps, blocks and SMs \cite{Ansorge:CUDAProg:CUP:2022,nextgenarch}.

\textbf{\em Warp}: A group of 32 cores constitute a ``warp-engine'' or a ``warp''. The threads in a given warp execute in lock-step running the same instruction at every clock-cycle. All the programs in a given warp maintain a single program counter executing in SIMD (Single Instruction Multiple Data) mode. In this work, we view a warp as a basic unit of execution with all the threads in a warp executing in a synchronous lock-step manner. 
In a scenario where different threads within a single warp need to execute different instructions, some threads may be idle (executing ``no-ops'') while others are actively executing instructions. This is to ensure that all the active cores are executing the same instruction. Additionally, if the total number of active threads in a warp is fewer than the number of available cores, the unused cores will remain idle, resulting in underutilization of the GPU's processing power.
}

\noindent{\bf Processes.} The system consists of a known static set of $n$ processes, $P_1, P_2, ... P_n$. These processes correspond to {\em threads} of a block of a Streaming Multiprocessor in a GPU. As discussed, these threads are distributed in {\em warps} (recall Fig.~\ref{fig:GPU-hier}). We denote by $p$ the size (number of processes) of a warp. Thus, the system consists of $\ceil{n/p}$ warps. In this paper, we focus on a single block of a single SM as our system (the circled components in \figref{SysModel}). 



\noindent{\bf \mpcolor{Phase-Based Computation.}}
The computation proceeds in {\em synchronized phases}. In a given phase, only one warp is scheduled by the {\em warp scheduler} (recall Section~\ref{sec:Intro}; also see below). Since there is one warp per phase in our model, the terms ``warp'' and ``phase'' inherently mean the same in this document. Within a phase, each process can perform a constant number of instructions, which include access to the shared objects along with some local computation. 
Recall from Section~\ref{sec:Intro} that processes (threads) in the same warp operate synchronously in a lock-step fashion, whereas processes across different warps operate asynchronously. Thus, processes in the same phase are synchronous and perform the same instructions (except Byzantine ones -- see below), but phase duration might be different. However, due to the small warp cycle approach of GPUs (cf. Section~\ref{sec:Intro}), the latency of the phases is bounded to a specific number of instructions. In other words, the scheduler enables processes in a warp to perform the same (small) number of instructions (even for the Byzantine ones).    



\begin{figure}[t]
\centering
\includegraphics[scale=0.37]{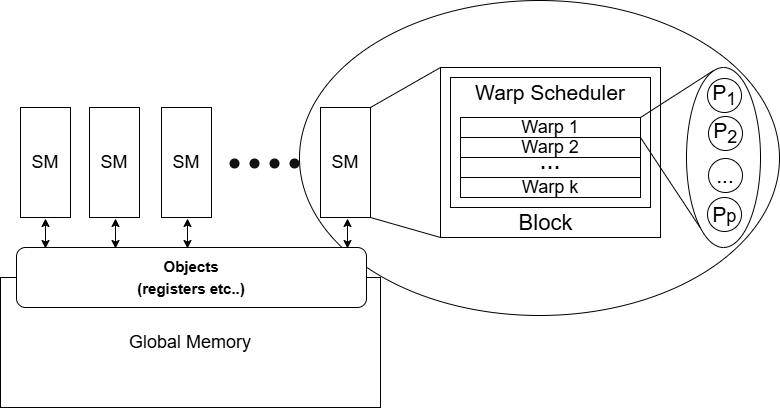}
\caption{Illustration of the system model}
\label{fig:SysModel}
\end{figure}

\ignore{
in the \figref{SysModel}. These processes operate in {\em synchronized phases}. The processes are the threads $t0,t1,t2..$ in \figref{gpuArch}, and a phase in our proposed model is one warp in the GPU system. In particular,
}

\noindent {\bf \mpcolor{Shared Memory System.}} 
Shared memory systems are characterized by a unified address space that is accessible to all components within the system~\cite{kshemkalyani:book:2011}. 
As typically assumed in the shared memory model, cf.~\cite{CKDISC2021}, the shared memory is non-corruptible and can only be accesses through objects like registers. Specifically, the communication between processes is through the APIs exposed by the objects in the system; processes invoke operations that in turn, return some response back to the process.

\noindent{\bf \mpcolor{Failure Model.}} We consider {\em Byzantine failures}~\cite{ByzGens}. In particular, we consider an adversary that may adaptively corrupt up to $f<n$ processes in the course of the computation. A corrupted process, called {\em Byzantine} or {\em faulty}, may deviate arbitrarily from the specified protocol (for example, instead of performing a read, it might choose to perform a write)~\cite{CKDISC2021}. More specific to our model, a faulty process within a warp in a given phase might perform a different set of instructions than the one indicated by the protocol. However, it cannot exceed the latency of the given phase or impersonate another process (and thus gain access to a phase in place of another process). A process that is not corrupted, called \textit{correct}, follows the protocol and takes infinitely many steps. We will be referring to a protocol that can tolerate up to $f$ Byzantine processes as ${f}${\em -resilient}. 

\noindent{\bf \mpcolor{Warp Scheduler.}} As mentioned, GPU architectures typically include a hardware scheduler to schedule the warps~\cite{Ogier:book:2013}. The scheduler maps each thread in the warp to an individual core, and then the threads start executing. {Threads (e.g., Byzantine) cannot masquerade as others to deceive the scheduler nor can they bypass the scheduler, as by design, processes/threads can only be activated for computation by the scheduler.} Motivated by existing hardware available to multi-core systems, in this work we consider {\em fair scheduling} \cite{wong:SIGOPS:2008,pabla:linux:2009}, which gives every process a chance to execute. Further, we assume that the warp scheduler does not exhibit Byzantine behavior, as otherwise it would not be fair anymore. The fair scheduler together with the synchronous execution of the threads inside the warp, realize the synchronized phase execution. 


As one might conclude, our proposed model abstracts the components and the operation of modern GPUs, as extensively analyzed in~\secref{Intro}. By establishing this model, we aim to explore its potential for solving fundamental problems in the field. To our knowledge, this is the first work that leverages a shared memory system model inspired by GPU architectures to address the Byzantine consensus problem, which we present next.

\section{Byzantine Consensus}
\label{sec:ByzIntro}
Generally speaking, in the consensus problem, a collection of $n$ processes propose values, and they need to agree on a value. Depending on the requirements imposed on the value agreed (validity property), we have different variations of the (Byzantine) consensus problem.
In this work, we consider two variations of Byzantine consensus:\\ 
(a) \emph{\cvbc}, where if all correct processes propose the same value, then this should be the one decided~\cite{attie:IPL:2002}.  \\(b)  \emph{\sbc}, where correct processes must agree on a value proposed by at least one correct process~\cite{Malkhi+:DC:2003,Bessani+:TPDS:2009}.

For completeness, we also briefly discuss \wbc \cite{Malkhi+:DC:2003}, although this is not the focus of this work.

\begin{definition}[\cvbc] 
	\label{def:consensus}
	Given a set of $n$ processes out of which $f<n$ might be faulty, and a set of values \spcolor{$V_c$} proposed by correct processes, the following must hold:
	\begin{itemize}
		\item \textbf{Termination: } Every correct process must eventually decide.
		\item \textbf{Agreement: } The value decided by correct process must be identical.
		\item \textbf{Common Validity: } If $V_c =\{v\}$, then $v$ must be the consensus value.
	\end{itemize}
\end{definition}

\begin{definition}[\sbc] 
	\label{def:strongconsensus}
	Given a set of $n$ processes out of which $f<n$ might be faulty, and a set of values $V_c$ proposed by correct processes, the following must hold:
	\begin{itemize}
		\item \textbf{Termination: } Every correct process must eventually decide.
		\item \textbf{Agreement: } The value decided by every correct process must be identical.
		\item \textbf{Strong Validity: } The consensus value must be proposed by at least one correct process, \ie final decision $\in V_c$.
	\end{itemize}
\end{definition}

\begin{wrapfigure}{r}{0.4\textwidth} 
\centering
 \includegraphics[width=\linewidth]{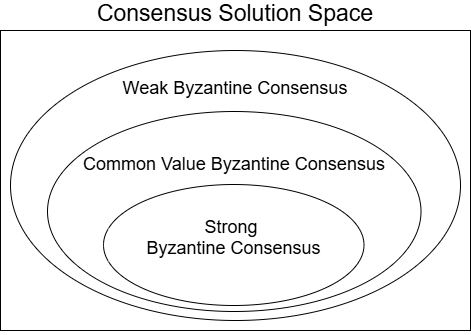}
\caption{Illustration of Consensus Solution Space}
\label{fig:Con-Space}
\end{wrapfigure}

A few remarks on this definition:\break (a) The set $V_c \subseteq V$, where $V$ is the set of all the values proposed, including those proposed by Byzantine processes. (b) Correct processes do not know the set $V$ (or $V_c$) in advance. (c) Malkhi {\em et al.}~\cite{Malkhi+:DC:2003} originally introduced \sbc in the context of binary consensus (decided value is either 0 or 1). Our definition generalizes it to support multi-valued consensus.

Malkhi {\em et al.}~\cite{Malkhi+:DC:2003} also considered \emph{weak consensus} in the context of binary proposal values, that is, the decided value is not necessarily one proposed by a correct process. This represents the most lenient approach, where, even if all correct processes propose the same value (e.g., 0), the consensus might still settle on a different value (e.g., 1) due to potential influence from faulty processes. This flexibility allows for the final decision to possibly be an input from a faulty process as well. 

\begin{definition}[\wbc] 
	\label{def:weakconsensus}
	Given a set of $n$ processes out of which $f<n$ might be faulty, and a set of values $V$ proposed by processes including the Byzantine ones, the following must hold:
	\begin{itemize}
		\item \textbf{Termination: } Every correct process must eventually decide.
		\item \textbf{Agreement: } The value decided by every correct process must be identical.
		\item \textbf{Validity: } The consensus value must be proposed by at least one process, \ie final decision $\in V$.
        
	\end{itemize}
\end{definition}

\figref{Con-Space} visually depicts the consensus solution space of the three types of Byzantine consensus mechanisms explored in this work. Observe that \sbc enforces strict criteria, ensuring that only values proposed by correct processes are eligible for the final decision. This makes \sbc the most reliable consensus method, as it guarantees that no erroneous values from faulty processes can influence the outcome. In contrast, \cvbc relaxes the validity requirement. When only a single unique value is proposed by correct processes, that value is chosen as the consensus value. However, if multiple values are proposed, the protocol does not mandate a specific choice (it could even be a non-proposed value), as long as the decision is consistent across correct processes.

\section{Byzantine Consensus Solution}
\label{sec:solns}
In this section, we present our solution for addressing both the \sbc and \cvbc. We first describe the design a novel object, called \abcas object (inspired from sticky bits), which is leveraged by the consensus algorithm. \smallskip

\noindent {\bf \BVCAS Shared Object.}
As we discussed in Section~\ref{sec:model}, the processes can access the shared memory via shared objects. In addition, the warp scheduler is scheduling the processes, in a fair manner, to execute their protocol, including accessing the memory. As argued by Malkhi {\em et al.}~\cite{Malkhi+:DC:2003}, in shared memory, and in the presence of Byzantine processes, a form of {\em persistent non-corruptible} object, like sticky bits, is needed to solve consensus. Otherwise, a value written by a correct process could later be overwritten by a Byzantine one, and hence leading to inconsistent states, which in turn can prevent reaching agreement. (Thus, simple registers or similar objects cannot be used.) 

Sticky bits~\cite{Malkhi+:DC:2003} are restricted to binary values (and hence useful for binary consensus). Since we consider multi-valued consensus, we needed to utilize a different object. In particular, we assume that the processes have access to a {non-corruptible}  shared object, which we call {\em Sticky Compare\&Swap}, or \BVCAS for short (combination of sticky bits and \cas object); through this object the processes can access the shared global memory. \smallskip

\noindent \textbf{\em Specification.}
\BVCAS{} maintains a {\em totally ordered list} of $n$ values, initially empty.  \sloppy{It supports two operations \Fwrite{($val$)} and \Fview{($len$)}.}
\begin{enumerate}   
    \item \Fwrite{($val$):} 
    \begin{enumerate}[label=\textnormal{(\alph*)}]
    \item It attempts to append $val$ to the list of values maintained if the maximum limit has not been reached.\label{spec:appendA} 
   \item It returns \textit{write successful} if $val$ is appended successfully; \textit{write failed} if append failed, and  \textit{limit reached} if $n$ values have been already appended. \label{spec:appendB}
    \end{enumerate}
    \item \Fview{($len$):} Returns the first $len$ values added to the list. 
\end{enumerate}    
\noindent \textbf{\em Properties.}
Recall that in our model, processes are grouped into warps, and in any given phase, only one warp is executing. Furthermore, (correct) processes within the same warp execute the same instructions. So, when multiple processes in the same phase invoke \Fwrite{()}, and provided the limit has not been reached, {\em only one} can succeed; all other processes should receive \mpcolor{\textit{failed}}. In other words, in any given phase, at most one process can change (append) the state of the list (and once a value is written, it cannot be overwritten or be deleted, as no such operations are available). \Fview, on the other hand is a concurrent read-only \mth. We elaborate more on these issues in Section~\ref{sec:correctness}.\smallskip  


\noindent{\bf Consensus Algorithm. } We proceed with the description of the consensus algorithm, Algorithm~\ref{alg:consensus}. The key issues that lead to its correctness are detailed in Section~\ref{sec:correctness}.   

\begin{algorithm}[t]
    \label{alg:consensus}
    \caption{Consensus Algorithm}
    \SetAlgoLined
    \textbf{DECIDE :} \\
    \tcc{On receiving first access to the \textit{\BVCAS} by the Scheduler}
    \Indp
    \Fwrite{($proposedValue$)} \label{lin:propose}\;
    \Indm
    \tcc{On receiving further accesses to the \textit{\BVCAS} by the Scheduler}
    \Indp
    ProposedList $\gets$ \Fview{($\ceil{n/p}$)}\;
    \Indm
    \Return $mode$(ProposedList)\;
    \label{lin:decide}
\end{algorithm}

Recall that the systems consists of $\ceil{n/p}$ warps ($p = $ warp size), and each warp execution is a phase in the model. Specifically, in every phase, the warp scheduler grants access to a different warp of $p$ processes to execute a fixed number of instructions from their protocol.  As we discussed later, the warp scheduler ensures that once a process is given access to execute, it will again be given access after $\ceil{n/p}$ phases (\ie after all other warps were granted access).

The consensus algorithm is designed to ensure agreement among processes using \BVCAS as a fundamental building block. Initially, each (correct) process, upon its first access to \BVCAS, attempts to write its proposed value using \Fwrite{} (Line~\ref{lin:propose}). This ensures that all participating processes compete to contribute their values to the shared memory structure. Subsequently, when a process gains further access to \BVCAS (regardless of the outcome of its \Fwrite{} operation), it retrieves a subset of previously written values using \Fview{}, which returns the first $\ceil{n/p}$ values appended in the list. The process then determines the consensus decision by computing the \textit{mode}\footnote{The \textit{mode} of a set \( S \) is defined as the element \( x \in S \) that occurs most frequently. In case of a tie, the smallest such element is chosen.
} of the retrieved values, selecting the most frequently occurring value. The use of \BVCAS ensures that all processes have access to the same sequence of written values, and by taking the mode, the system converges to a single consensus value.


Conceptually the idea seems simple, but there are a few subtle issues that need to be examined carefully in order to show that this indeed solves Byzantine Consensus. We do so in the next section.

\section{Correctness and Fault-Resilience of the Solution}
\label{sec:correctness}

In order to show that Algorithm~\ref{alg:consensus} solves both versions of Byzantine consensus we consider in this work (cf. Section~\ref{sec:ByzIntro}), we impose specific assumptions on top of the model presented in Section~\ref{sec:model}. We justify these assumptions based on existing GPU architectures. The reason we have not imposed these assumptions in the model is because we want to keep the model more general (to be used in other fundamental problems); these assumptions are needed for the specific solution we provide. (We do not claim that these specific assumptions are necessary to solve consensus, they are however, sufficient for the consensus solution we provide.)  

\noindent {\bf \mpcolor{Specific Assumptions}.}
To solve Byzantine consensus we further assume:\smallskip

\noindent
\textbf{(a) Round Robin (RR) Warp Scheduler.} GPUs are designed with an internal warp scheduler to efficiently manage core utilization and prevent idling or resource wastage. This scheduler ensures that computational resources are allocated effectively, maintaining high throughput. Recall that in our proposed model (Section~\ref{sec:model}), we assume that the scheduler functions as a fair scheduler, meaning it treats all warps equitably and is not influenced by external disruptions. Furthermore, we assume that Byzantine processes cannot affect the scheduler’s behavior. In this section, additionally, we assume that the scheduler follows a {\em round-robin} scheduling strategy, distributing phases evenly among all available warps where each warp is scheduled again, only after all other warps have been scheduled. We believe this is a reasonable and realistic assumption, as Round-robin in GPUs is a baseline scheduling approach that chooses warps according to their warp ID in either increasing or decreasing order~\cite{Jeon:springer:2025}. Multiple improvements have been proposed and implemented over the baseline scheduling approach for efficient utilization of the GPU cores.\smallskip

\noindent
\textbf{(b) \abcas Operations within a Warp Phase.} As detailed in \secref{model}, modern GPUs schedule a constant number of instructions before switching warps \cite{NVIDIA:2024}. Here, we make the latency of a phase more specific with respect to the \abcas operations. Since \Fwrite{($val$)} appends to the end of the list, it requires a constant number of instructions regardless of the list's length. Thus, we assume that this operation fits within a single phase, and for simplicity, no more than one \Fwrite{} can occur per phase (otherwise, the termination bounds shown in the forthcoming analysis could be adjusted accordingly). However, \Fview{($len$)} may span multiple phases due to the list’s variable size. Specifically, we set $r\leq n$ to be the number of values that \Fview{} can reads within a phase. This means that reading $len$ values takes $\ceil{n/p} \cdot \ceil{len/r}$ phases. Given current warp scheduling behavior, "this assumption aligns with the working of modern hardware architectures.\smallskip

\noindent
\textbf{(c) Global Memory Access.} In our solution, we assume that the global memory isn’t directly accessible to processes. Instead, it can only be reached through specific objects, like the sticky bits~\cite{Malkhi+:DC:2003} or the \abcas object we’ve designed. These objects serve as controlled gateways to the global memory, ensuring that all interactions go through well-defined methods. The APIs of these objects cannot be changed, meaning processes have to work within the constraints of the existing methods to read or modify the state. This approach helps maintain consistency and control over how the global memory is used.

\noindent With these assumptions in place, we are now ready to prove the correctness and resilience of our consensus solution.

\begin{table}[]
\centering
\color{black}
\begin{tabular}{|>{\centering\arraybackslash}m{2cm}|>{\centering\arraybackslash}m{10cm}|}
\hline
\textbf{Parameter} & \textbf{Description} \\ 
\hline
\textbf{$n$} & The total number of processes in the system. \\ 
\hline
\textbf{$V$} & The set of proposed values from all processes including the faulty ones. \\ 
\hline
\textbf{$V_c$} & The set of proposed values by correct processes. \\ 
\hline
\textbf{$f$} & The maximum Byzantine processes that the system can tolerate. \\ 
\hline
\textbf{$p$} & The number of processes assigned per warp or phase. \\ 
\hline
\textbf{$r$} & The number of memory elements that can be read in a given phase.\\  
\hline
\end{tabular}\smallskip
\caption{System Parameters}
\label{table:par}
\end{table}


\noindent
{\bf Correctness proofs.} For convenience, we provide Table~\ref{table:par} that summarizes the main system parameters.
In \thmref{scon} we prove that \sbc is solved within $\Theta(n^2/p^2r)$ phases, provided that \spcolor{$f< \frac{n}{(|V_c|+1)p}$}. It follows that the same algorithm solves \cvbc within the same number of phases, with \spcolor{$f<\ceil{\frac{n}{2p}}$, since $|V_c|=1$}.
\begin{theorem}\label{thm:scon} 
\algoref{consensus} solves the \sbc problem within $\Theta(n^2/p^2r)$ phases, for \spcolor{$f< \frac{n}{(|V_c|+1)p}$}.
\end{theorem}
\begin{proof}

For \sbc, we want the decided consensus value to be proposed by at least one correct process. We show that Agreement, Strong Validity, and Termination properties follow essentially from the properties of the \BVCAS object and the RR Warp Scheduler.\smallskip 

\noindent
\textbf{\em Properties of \BVCAS and RR Scheduler:} \BVCAS ensures that in a given phase, if $p$ processes attempt to invoke \Fwrite{} simultaneously, then only one will succeed, while the others will fail (see \Fwrite.\ref{spec:appendB}). Therefore, if a Byzantine process executes \Fwrite{} in a phase, it can potentially succeed, \ie \emph{win} the \Fwrite{}, and append a value $v$. The RR scheduler ensures that a second access to a process is only allowed after all processes have completed their first access. Thus, if process $P_i$ executes an \Fwrite{} request in phase $\ell$, it will be able to invoke \Fview{($\ceil{n/p}$)} in phase $\ell+\ceil{n/p}$.\smallskip

\noindent
\textbf{\em Agreement:} All the \Fwrite{} calls in a single phase try to append to the list in a single position as they execute in lock step manner. This is the position in the array next to the last appended proposed value in the \BVCAS object. Once appended, the value can not be modified and access to the memory is allowed only through the \BVCAS object. Therefore \BVCAS guarantees that $ProposedList$ is consistent across all processes. So all correct processes will read the same list, compute the same mode value, agree on the same value and thus, fulfilling the agreement property. \smallskip

\noindent
\textbf{\em Strong Validity:}
Given that there are $f$ Byzantine processes, in the worst-case scenario, all Byzantine processes can win the \Fwrite{()} in $f$ phases. Consequently, the maximum number of phases with Byzantine presence is $f$ out of the $\ceil{n/p}$. The remaining $\ceil{n/p} - f$ values are proposed by correct processes. We consider the following cases.

\noindent \texttt{Case (a): The Byzantine processes propose NULL or an empty value.} In this case the $mode( ProposedList )$ results in a value proposed only by correct processes.\\
\spcolor{
\noindent \texttt{Case (b): All Byzantine processes propose a different value.} In this case the values proposed by them will never gain majority. Consider a subsequent case where the values proposed by correct processes are also different, that is, $ProposedList$ contains $\ceil{n/p}$ different values. So $|V|= \ceil{n/p}$, as there are $\ceil{n/p}$ unique values. Substituting $\ceil{n/p}$ for $|V|$, in $f< \frac{n}{(|V_c|+1)p}$ gives {$f< \frac{|V|}{(|V_c|+1)}$}. Given that in the set $V$ at most $f$ malicious values can be proposed, in this scenario $|V|-|V_c|=f$.
Substituting \( |V| = f + |V_c| \) into the inequality, we get $f (|V_c| + 1)< f + |V_c|$. Solving this for $f$ gives $f<1$. Thus, a solution is possible only if there are no faulty processes, or at least two correct processes propose the same value (or a correct and a faulty propose the same value). In all these cases, validity is satisfied. }\\
\noindent  \texttt{Case (c): The Byzantine processes collude on appending the same value $v' \notin V_c$.} 
To consider the worst case, assume that the correct processes append values different from $v'$, \ie we have $(\ceil{n/p}) - f$ values from $V_c$. Even if the values appended by the correct processes are equally divided across all the values in $V_c$, \ie $\frac{(\ceil{n/p}) - f}{|V_c|}$, $v'$ can only be chosen when $f > \frac{(\ceil{n/p}) - f}{|V_c|}$. This would mean that $f \geq \frac{n}{(|V_c|+1)p}$, which contradicts our initial assumption. Therefore, the value proposed by the Byzantine processes can only be chosen if one or more correct processes also propose $v'$, making it the result of $mode( ProposedList )$. This ensures the validity of \sbc.

\noindent \texttt{Case (d): Byzantine processes proposing values from $V_c$.}
In this case, Byzantine processes do not introduce any new values but instead propose values already present in $V_c$, the set of values proposed by correct processes. The goal of Byzantine processes in this case could be to manipulate the final decision by influencing the frequency of certain values in the $ProposedList$. Since the final decision is based on $mode(ProposedList)$, a Byzantine process can only successfully influence the outcome if it can make a faulty-preferred value the majority. However, because there are at most $f$ Byzantine processes and at least $ (\ceil{n/p}) - f $ correct values, their ability to shift the mode is limited. If correct processes propose values uniformly across $V_c$, the most frequent value will always be one proposed by correct processes unless the Byzantine processes coordinate with some correct processes to propose the same value. Even in the worst case where Byzantine processes try to bias the mode toward a specific correct value, the chosen value will still belong to $V_c$, ensuring strong validity.
 
\noindent  \texttt{Case (e): A mix of strategies by Byzantine processes.} 
In this case, Byzantine processes do not follow a single strategy but instead use a combination of the previous approaches. Since we have already shown that in each individual case the Byzantine processes cannot violate validity, a mixture of these strategies will not allow them to break strong validity either.

Thus, no matter what strategy or combination of strategies the Byzantine processes employ, they cannot force an invalid value to be chosen, and strong validity is preserved with the claimed bound on $f$.\smallskip

\noindent
\textbf{\em Termination:}
In our system model, during any given phase, each process performs a fixed number of instructions. An \Fwrite{} operation requires a constant number of instructions and will complete within the phase. 
With a RR scheduler, each process's \Fwrite{} invocation will occur once within the first $\ceil{n/p}$ phases. The \Fview{($len$)} operation might span multiple phases to complete and will take $\ceil{n/p} \cdot \ceil{len/r}$ phases for each process. Combining \Fwrite{} and \Fview{}, and given that $len=\ceil{n/p}$ (as per \algoref{consensus}), all correct processes will decide on the same value within $\Theta(n^2/(p^2r))$ phases, provided that \spcolor{ $f < \frac{n}{(|V_c|+1)p}$}. This completes the proof. \qed 
\end{proof}
From \thmref{scon} we get the following result for \cvbc.
(We have illustrated in the Byzantine consensus solution space \figref{Con-Space}, that any algorithm solving \sbc will solve \cvbc.)

\begin{corollary}
\spcolor{\algoref{consensus} solves the \cvbc problem within $\Theta(n^2/p^2r)$ phases, for $f< \ceil{\frac{n}{2p}}$.}
\end{corollary}

\noindent {\bf \mpcolor{Remarks on Resilience}.}
Observe that for \sbc the resilience depends on $|V_c|$: as the size of $V_c$ increases, the resilience decreases, and when $|V_c|+1 \geq \ceil{n/p}$, the algorithm is no longer resilient. It would be interesting to investigate whether this is not only a sufficient, but also a necessary condition within the system model we consider. 

Also observe that the resilience depends on $p$, the number of processes that the scheduler permits access to \BVCAS in a given phase. The lower $p$ is, the higher resilience we get. In the case that $p=1$ and $|V_c|=1$, $f$ can be as large as $\ceil{n/2}$ for \sbc and \cvbc.

\mpcolor{It is essential to highlight that the current analysis adopts a deliberately pessimistic perspective. This work is an initial step toward understanding the model’s resilience under worst-case adversarial conditions. We assume that a single Byzantine thread in a phase always succeeds in writing its own value,  while honest threads fail, which is an unrealistic scenario that likely underestimates practical resilience. A more optimistic probabilistic fault model, where, e.g., all threads have equal chances of success, could reveal significantly higher resilience.}

\section{Related Work}
\label{sec:related-sec}
In this work, we aim to formalize the model of GPUs with Byzantine faults and explore the problems that can be solved within this framework. Consensus is one of the problems that is being extensively studied in message passing~\cite{miller:SIGSAC:2016,gilad+:SOSP:2017,Liu+:IEEETransactions:2019} and shared memory~\cite{attie:IPL:2002,Malkhi+:DC:2003,10.1007/s00446-005-0125-8}. Diverse researchers have approached the Byzantine consensus problem in shared memory systems in various ways. Malkhi {\em et al.}~\cite{Malkhi+:DC:2003} demonstrated that Byzantine-tolerant objects in shared memory can be constructed with $f < n/3$ tolerance, using access control lists, persistent objects with defined fault tolerance limits, and redundancy. Sticky bits were utilized by Alon {\em et al.} \cite{10.1007/s00446-005-0125-8} to solve \sbc among $n$ processes where $n \geq 3f+1$, $f$ being the number of faulty processes. Attie \cite{attie:IPL:2002} showed that weak Byzantine consensus can be achieved, but only by using non-resettable or sticky shared objects even in a reliable shared memory setting. Attie \cite{attie:IPL:2002} has also detailed an impossibility proof that shows that even with limited-access restriction to Byzantine processes, it is impossible to achieve Byzantine consensus. However, we have not been able to find any implementation of sticky bits in hardware. The closest available technology, Write-Once-Read-Many (WORM) \cite{worm:techtarget:2022}, lacks concurrency support. 

Construction of shared memory objects like Byzantine tolerant single writer multi reader (SWMR) registers have been extensively explored over the years~\cite{CKDISC2021,Malkhi+:DC:2003,hu+:DC:2024}. Hu and Toueg  have recently proposed implementation of SWMR registers from SWSR registers \cite{hu+:DC:2024} using signatures. A new computing paradigm where shared memory objects are protected by fine-grained access policies was introduced by Bessani {\em et al.}~\cite{Bessani+:TPDS:2009}. Their research highlights the need for byzantine tolerant protocols and objects for shared memory specifically as such solutions have been long prevalent in message passing systems. Construction of a shared memory object called Policy-Enforced Augmented Tuple Space (PEATS) is also detailed in their work. Their proposed model has a resistance of $n \geq 3f+1$.

These studies underscore the necessity for Byzantine-tolerant protocols and objects in shared memory systems. Note that most research assumes security premises such as Public Key Infrastructure (PKI) or controlled access, which are complex to implement and design. We build on these works to design a system framework that leverages existing hardware to tolerate Byzantine faults.

\section{Conclusion}
Our research aims to leverage existing hardware and software capabilities to address the challenges posed by Byzantine faults. In this paper, we formalize a model of shared memory systems inspired by GPU architectures, which can handle processes exhibiting Byzantine faults.

Our contributions include demonstrating the feasibility of achieving \sbc and \cvbc within this model by utilizing a concurrent object, called \BVCAS. In Appendix~\ref{app}, the interested reader can find a crash-tolerant implementation of \BVCAS in a non-corruptible shared memory setting; the implementation builds on Compare\&Swap (CAS) objects, which are known to be implemented in hardware.

One might wonder why we cannot use a CAS object in place of \BVCAS. CAS is used for accessing a specific memory location. For the needs of our consensus algorithms, we want to ensure that once a value is successfully written in a memory location, it cannot be modified. In the presence of a Byzantine process, it seems that CAS cannot ``protect'' a memory location from not being updated again.This requirement of memory protection is facilitated by \bvcas.

We believe our work opens new opportunities for the Distributed and Parallel Computing community to explore. Our abstraction focused on a single block of a single SM, as multiple schedulers may operate across different blocks. Future work could investigate synergies between warp schedulers for interblock computations. Additionally, we aim to enhance the understanding and modeling of GPU behavior under Byzantine faults, by considering other fundamental problems of distributed computing within the GPU-inspired model we have devised. An interesting question is whether the resilience achieved in this work is optimal.\medskip

\noindent
\textbf{Acknowledgments: }We would like to thank  Gadi Taubenfeld and anonymous reviewers for their insightful comments.

\bibliographystyle{splncs04} 
\bibliography{citations} 

\begin{thebibliography}{10}
\providecommand{\url}[1]{\texttt{#1}}
\providecommand{\urlprefix}{URL }
\providecommand{\doi}[1]{https://doi.org/#1}

\bibitem{10.1007/s00446-005-0125-8}
Alon, N., Merritt, M., Reingold, O., Taubenfeld, G., Wright, R.N.: {Tight bounds for shared memory systems accessed by Byzantine processes}. Distrib. Comput.

\bibitem{attie:IPL:2002}
Attie, P.: {Wait-free Byzantine consensus}. Information Processing Letters  \textbf{83}(4),  221--227 (2002)

\bibitem{Bessani+:TPDS:2009}
Bessani, A.N., Correia, M., da~Silva~Fraga, J., Cheuk~Lung, L.: Sharing memory between byzantine processes using policy-enforced tuple spaces. IEEE Transactions on Parallel and Distributed Systems  \textbf{20}(3),  419--432 (2009)

\bibitem{CiniYalcin:GPU-CPU:ACMSurvey:2020}
Cini, N., Yalcin, G.: {A Methodology for Comparing the Reliability of GPU-Based and CPU-Based HPCs}. ACM Comput. Surv.  \textbf{53}(1) (feb 2020)

\bibitem{CKDISC2021}
Cohen, S., Keidar, I.: Tame the wild with byzantine linearizability: Reliable broadcast, snapshots, and asset transfer. In: DISC 2021. pp. 18:1--18:18 (2021)

\bibitem{NVIDIA:2024}
Corporation, N.: {CUDA C++ Programming Guide} (2024), \url{https://docs.nvidia.com/cuda/cuda-c-programming-guide/}, accessed: 2025-03-12

\bibitem{gilad+:SOSP:2017}
Gilad, Y., Hemo, R., Micali, S., Vlachos, G., Zeldovich, N.: {Algorand: Scaling byzantine agreements for cryptocurrencies}. In: 26th SOSP. pp. 51--68 (2017)

\bibitem{Maurice:1991:ACMTrans}
Herlihy, M.: Wait-free synchronization. ACM Trans. Program. Lang. Syst.  \textbf{13}(1),  124–149 (Jan 1991). \doi{10.1145/114005.102808}

\bibitem{HerNir:AMP:Book:2012}
Herlihy, M., Shavit, N.: {The Art of Multiprocessor Programming}. Morgan Kaufmann Publishers Inc., San Francisco, CA, USA, 1st edn. (2012)

\bibitem{HW1990}
Herlihy, M., Wing, J.M.: {Linearizability: {A} Correctness Condition for Concurrent Objects}. {ACM} Trans. Program. Lang. Syst.  \textbf{12}(3),  463--492 (1990)

\bibitem{hu+:DC:2024}
Hu, X., Toueg, S.: On implementing {SWMR} registers from {SWSR} registers in systems with byzantine failures. CoRR  \textbf{abs/2207.01470} (2022)

\bibitem{Huang+:SDC-HPC:SC:2022}
Huang, Y., Guo, S., Di, S., Li, G., Cappello, F.: {Mitigating Silent Data Corruptions in HPC Applications across Multiple Program Inputs}. In: SC22. pp. 1--14 (2022)

\bibitem{Jeon:springer:2025}
Jeon, H.: GPU Architecture, pp. 531--559. Springer Nature Singapore (2025)

\bibitem{david+:PMPP:science:2017}
Kirk, D.B., mei W.~Hwu, W.: Chapter 4 - memory and data locality. In: Programming Massively Parallel Processors (Third Edition), pp. 71--101 (2017)

\bibitem{Kolesnichenko+:2017+GPCE}
Kolesnichenko, A., Poskitt, C.M., Nanz, S.: {SafeGPU: Contract- and library-based GPGPU for object-oriented languages}. Computer Languages, Systems \& Structures  \textbf{48} (2017), special Issue on the 14th International Conference on GPCE

\bibitem{KMPOPODIS2023}
Kowalski, V., Most{\'{e}}faoui, A., Perrin, M.: Atomic register abstractions for byzantine-prone distributed systems. In: OPODIS 2023 (2023)

\bibitem{kshemkalyani:book:2011}
Kshemkalyani, A.D., Singhal, M.: Distributed Computing: Principles, Algorithms, and Systems. Cambridge University Press, 1st edn. (2011)

\bibitem{ByzGens}
Lamport, L., Shostak, R., Pease, M.: {The Byzantine Generals Problem}. ACM Trans. Program. Lang. Syst.  \textbf{4}(3),  382–401 (jul 1982)

\bibitem{Liu+:IEEETransactions:2019}
Liu, J., Li, W., Karame, G.O., Asokan, N.: Scalable byzantine consensus via hardware-assisted secret sharing. IEEE Transactions on Computers  \textbf{68}(1) (2019)

\bibitem{Ogier:book:2013}
Maitre, O.: Understanding nvidia gpgpu hardware. In: Massively Parallel Evolutionary Computation on GPGPUs, pp. 15--34 (2013)

\bibitem{Malkhi+:DC:2003}
Malkhi, D., Merritt, M., Reiter, M.K., Taubenfeld, G.: {Objects shared by Byzantine processes}. Distrib. Comput.  \textbf{16}(1),  37–48 (feb 2003)

\bibitem{miller:SIGSAC:2016}
Miller, A., Xia, Y., Croman, K., Shi, E., Song, D.: {The honey badger of BFT protocols}. In: Proceedings of the 2016 ACM SIGSAC conference on computer and communications security. pp. 31--42 (2016)

\bibitem{Olmedo+:RTAS:IEEE:2020}
Olmedo, I.S., Capodieci, N., Martinez, J.L., Marongiu, A., Bertogna, M.: Dissecting the cuda scheduling hierarchy: A performance and predictability perspective. In: IEEE RTAS 2020. pp. 213--225 (2020)

\bibitem{pabla:linux:2009}
Pabla, C.S.: {Completely fair scheduler}. Linux Journal  \textbf{2009}(184), ~4 (2009)

\bibitem{worm:techtarget:2022}
Sheldon, R.: Worm (write once, read many) (2022), \url{https://www.techtarget.com/searchstorage/definition/WORM-write-once-read-many}

\bibitem{Weaver+:SC24:2024}
Weaver, A., Kavi, K., Milojicic, D., Enriquez, R.P.H., Hogade, N., Mishra, A., Mehta, G.: Granularity- and interference-aware gpu sharing with mps. In: SC24-W: Workshops of SC24. pp. 1630--1637 (2024)

\bibitem{wong:SIGOPS:2008}
Wong, C.S., Tan, I., Kumari, R.D., Wey, F.: {Towards Achieving Fairness in the Linux scheduler}. ACM SIGOPS Operating Systems Review  \textbf{42}(5),  34--43 (2008)

\end{thebibliography}

\appendix
\section{Appendix}
\label{app}
\subsection{Algorithm}
\label{sec:abcas-imple}
\textcolor{black}{In this section, we present a \emph{wait-free crash-tolerant} implementation of \BVCAS. As stated in \secref{solns}, we assume that global memory is accessible only through predefined objects and that object APIs cannot be modified by Byzantine processes. Under this assumption, a crash-tolerant object is sufficient to achieve Byzantine consensus in the proposed model.}

\subsection{Preliminaries}
\label{sec:app-prelim}

\noindent We first provide some basic concepts that will help us understand the \BVCAS implementation.

\vspace{1mm}
\noindent
\textbf{Crash-tolerant object: }A \emph{crash failure} is a special case of a Byzantine failure, where the faulty component stops taking any further steps indefinitely; until that point it always follows the specified protocol. A concurrent object is \emph{crash-tolerant} if it is able to satisfy the object specification despite threads invoking it crash midway.  

\vspace{1mm}
\noindent
\textbf{Atomic Object: }A concurrent object/instruction is \emph{atomic} if all the \mth{s} of the object appear to take effect instantaneously at some point between the method’s invocation and response \cite[Chap. 3]{HerNir:AMP:Book:2012}, \cite{HW1990}.



\vspace{1mm}
\noindent
\textbf{Wait-Freedom: }A concurrent object is \emph{\wf} if a non-crashed process is able to complete all its \mth{s} in finite number of its own steps regardless of execution of other processes that could potentially crash. 

\vspace{1mm}
\noindent
\textbf{Compare and Swap (CAS): }It is an atomic wait-free hardware instruction supported by several modern day architectures. A \cas object is a wrapper around the hardware \cas instruction. Consider a \cas object $O$ which is initialized with a value $v$. The \cas object supports two \mth{s}: (a) $O.\CAS(\varepsilon, \nu)$: This \mth atomically compares the current value of $O$ with the expected value $\varepsilon$. If it is same, then $\varepsilon$ is replaced with $\nu$ in $O$ atomically. And this \mth returns true. If the contents of $O$ are not same as $\varepsilon$, then $O$ remains unchanged and the \mth returns false. (b) $O.read()$: This \mth returns the current value of $O$. 




\subsection{The \BVCAS Algorithm}
\label{subsec:Linz-BVCAS}

\algoref{synABCAS} provides a crash-tolerant
implementation of \BVCAS object. This implementation is designed to function effectively in a setting composed of a collection of crash-prone synchronous processes, where at least one does not crash (otherwise it would not have been wait-free). We have utilized the suffix $``\_sh"$, to indicate that the variable is a shared variable across the procedures of the \BVCAS object.

\vspace{1mm}
\noindent 
\textbf{Data-Structure.} For \BVCAS implementation, we are maintaining an array of $n$ atomic registers which are initialized with NULL values. \spcolor{The memory complexity of this \BVCAS implementation is $O(n/p)$, in cases where $p=1$, the complexity is $O(n)$}

\vspace{1mm}
\noindent 
\textbf{Algorithm Details. } We now detail the two main procedures of the implementation.\smallskip

\noindent \Fwrite{($val$)}: In this procedure, all processes within a phase read the counter value ($k$), indicating the array position where they will attempt to write their proposed value. At \Lineref{check} in \algoref{synABCAS}, we verify if $k \leq$ n, signifying that the upper limit for proposals has been reached, disallowing further proposals. If the limit isn't reached, at \Lineref{casTwo}, all correct processes attempt to increment the counter for the next phase, using CAS, to ensure that this phase's value isn't overwritten. In \Lineref{casOne}, all correct processes will attempt to write their proposed value using CAS. 
From the code and the fact that CAS is  wait-free, it follows that this  implementation is also wait-free (see Observation~\ref{ob:abcas-wf}). Therefore, even if all but one processes have crashed, progresses is still achieved and the proposed value will be appended.\smallskip

\noindent\Fview{($len$)}: This procedure reads values from the consensus data structure. It takes a parameter $len$ indicating the number of values to read from the beginning of the data structure. It returns a list of values from $proposedList$ indexed from $1$ to $len$.\smallskip


\begin{algorithm}[t]
    \caption{Crash-tolerant Implementation of \textbf{\BVCAS}}
    \label{alg:synABCAS}
    \SetAlgoLined
    \SetKwFunction{Constructor}{Constructor}
    \SetKwProg{Class}{class}{:}{end}
    \SetKwFunction{Fwrite}{StickyCAS-append}
    \SetKwFunction{Fview}{StickyCAS-read}
    \SetKwComment{Comment}{// }{}
    \SetKwData{Counter}{counter}
    \SetKwData{N}{N}
    \SetKwData{Status}{status}
    \SetKwData{Result}{result}
    \SetKwProg{Fn}{{Procedure}}{:}{\textbf{end}}
    \SetKwFunction{CAS}{CAS}
    \Class{StickyCAS}{
        \BlankLine
        \Fn{\boldmath$\Constructor${\KwSty{(int n)}}}{
            \tcp{\textbf{Initialize proposedList\_sh and counter\_sh}}
            int \Counter\_sh $\gets 1$  \;
            \tcp{\textbf{here `value' is the datatype of the proposed values}}
            value proposedList\_sh[n] $\gets$ NULL\;
                    \label{lin:init}
        }
        \BlankLine
        { Status }\Fn{\boldmath$\Fwrite${\KwSty{(val)}}}{
            int k $\gets \Counter\_sh$ \; 
            value m $\gets$ proposedList\_sh[k] \;
            \If{($\mathrm{k} \leq$ n)}{
            \label{lin:check}
                counter\_sh.\CAS(k, k+1)\;
                \label{lin:casTwo}
                \tcp{\textbf{Try to insert the value using CAS}}
                \If{(proposedList\_sh[k].CAS(m, val))}
                {
                \label{lin:casOne}
                \textbf{\tcp{Check if already n values are inserted}}
                    \Return ``Append Success''\; 
                    }
                \Else{
                    \Return ``Append Failed''\; 
            }
            }
            \Return ``Limit Reached''\; 
           }
        \BlankLine
        { value[ ] } \Fn{\boldmath$\Fview${\KwSty{(len)}}}{
            \Result[ ] $\gets$ proposedList\_sh[1 to  len]\;
            \Return \Result\;
        }
    }

\end{algorithm}

\noindent We now show the correctness of \algoref{synABCAS}.


\begin{theorem}
	\label{thm:abcas-spec}
	The implementation shown in \algoref{synABCAS} satisfies the specification of an \abcas object described in \secref{model}. 
\end{theorem}

\begin{proof}
	To prove this, we will iterate through each of the items mentioned as a part of the specification of \abcas. We will start with the \Fwrite \mth. 
\textcolor{black}{
	\begin{itemize}
		\item \Fwrite.\ref{spec:appendA}: This specification states that the \Fwrite{$val$} function attempts to add the proposed value $val$ to the list of values if the maximum limit has not been reached. This is ensured by the check on \lineref{check}. 
		\item \Fwrite.\ref{spec:appendB}: This specification states that the \Fwrite{$val$} returns \textit{write successful} when the process has appended successfully, \textit{write failed} if append failed, and  \textit{limit reached} if $n$ values have been already appended. This follows from the behaviours of the \cas object and the return values as shown in the code (Lines 13, 16 and 19).        
        \item In \secref{solns}, we assume that the function call completes within a single phase and that only one invocation of \Fwrite{} occurs per phase. This assumption is supported by the code. The \Fwrite{$val$} \mth executes one read operation and two \cas{} operations on shared memory, amounting to a constant number of instructions. As a result, these operations can fit within a single warp cycle. Given by the short warp cycles property of the model, we can be assured that no process can execute more than one \Fwrite{} per phase.
        \end{itemize}
}
	\begin{itemize}
		\item \Fview: This specification states that the function returns the first $len$ values appended in the list. This follows from the atomic nature of $proposedList\_sh[]$ array. This \mth scans the first $len$ elements and returns them while a concurrent \Fwrite \mth only appends. Hence concurrent \Fwrite and \Fview can be linearized. 
	\end{itemize}
This completes the proof. 
\end{proof}

Having shown the correctness of \algoref{synABCAS} for implementing the \abcas object, we next observe that the implementation of \abcas is \wf. From inspection of the algorithm, one can see that there are no loops. The \abcas implementation uses the \cas object which is \wf. Hence, the following observation follows. 

\begin{observation}
	\label{ob:abcas-wf}
	The implementation shown in \algoref{synABCAS} is \wf. 
\end{observation}


\noindent{\bf Using \bvcas object instead of \cas object.} One might wonder why we cannot use a CAS object in place of \BVCAS{} in \algoref{consensus} (Section~\ref{sec:solns}). As mentioned in \secref{app-prelim}, CAS is used for accessing a specific memory location. For the needs of Algorithm~\ref{alg:consensus}, we want to ensure that once a value is successfully written in a memory location, it cannot be modified. In the presence of a Byzantine process, it seems that CAS cannot ``protect'' a memory location from not being updated again (at a later phase of the computation -- Byzantine processes can deviate from the protocol and modify a previously written value). This requirement of memory protection is facilitated by \bvcas.

\end{document}